\documentclass[aip, numerical, jmp]{revtex4-1}
\usepackage{amssymb}
\usepackage{amsmath}
\usepackage{amsfonts}
\usepackage{bbm}
\usepackage{amsthm}
\usepackage{mathrsfs}
\usepackage{marvosym}
\usepackage{color}
\usepackage[all,cmtip]{xy}
\usepackage[utf8]{inputenc}

\usepackage{enumitem}

\usepackage{upgreek}

\theoremstyle{plain}
\newtheorem{theo}{Theorem}[section]
\newtheorem{lem}[theo]{Lemma}
\newtheorem{propo}[theo]{Proposition}
\newtheorem{cor}[theo]{Corollary}

\theoremstyle{definition}
\newtheorem{defi}[theo]{Definition}
\newtheorem{ex}[theo]{Example}

\newtheorem{rem}[theo]{Remark}

\numberwithin{equation}{section}

\def\nn{\nonumber}

\def\bbR{\mathbb{R}}
\def\bbC{\mathbb{C}}
\def\bbN{\mathbb{N}}
\def\bbZ{\mathbb{Z}}

\def\ii{{\,{\rm i}\,}}

\def\id{\mathrm{id}}

\def\1{\mathbbm{1}}

\newcommand\mycom[2]{\genfrac{}{}{0pt}{}{#1}{#2}}

\def\OO{\mathcal{O}}
\def\JJ{\mathcal{J}}
\def\DD{\mathcal{D}}

\def\Pol{\mathcal{P}^\ast}
\def\TT{\mathcal{T}}

\def\res{\mathrm{res}}

\def\sk{\vspace{2mm}}

\begin{document}
	
	\title{%
		Wavefront sets and polarizations on supermanifolds
	}
	
	\author{Claudio Dappiaggi}
	\email{claudio.dappiaggi@unipv.it}
	\affiliation{Dipartimento di Fisica Nucleare e Teorica, 
	Universit\'a degli Studi di Pavia \& INFN, Sezione di Pavia, Via Bassi, 6, 27100 Pavia, Italy.}
	\author{Heiko Gimperlein}
	\email{h.gimperlein@hw.ac.uk}
	\affiliation{Department of Mathematics, Heriot-Watt University, Colin Maclaurin Building, 
	Riccarton, Edinburgh EH14 4AS, United Kingdom.}
	\affiliation{Maxwell Institute for Mathematical Sciences, Edinburgh, United Kingdom.}
	\affiliation{Institut f\"{u}r Mathematik, Universit\"{a}t Paderborn, Warburger Str.\ 100, 33098 Paderborn, Germany.}
	\author{Simone Murro}
	\email{simone.murro@ur.de} 
	\affiliation{Fakult\"at f\"ur Mathematik, Universit\"at Regensburg, Universit\"atsstra\ss e 31, 93040 Regensburg, Germany.}
	\author{Alexander Schenkel}
	\email{aschenkel83@gmail.com}
	\affiliation{Department of Mathematics, Heriot-Watt University, Colin Maclaurin Building, 
	Riccarton, Edinburgh EH14 4AS, United Kingdom.}
	\affiliation{School of Mathematical Sciences, University of Nottingham, University Park, Nottingham NG7 2RD, United Kingdom.}

	\date{\today}


\begin{abstract}
\noindent In this paper we develop the foundations for microlocal analysis on supermanifolds. Making use of pseudodifferential operators on supermanifolds as introduced by Rempel and Schmitt, we define a suitable notion of super wavefront set for superdistributions which generalizes Dencker's polarization sets for vector-valued distributions to supergeometry. In particular, our super wavefront sets detect polarization information of the singularities of superdistributions. We prove a refined pullback theorem for superdistributions along supermanifold morphisms, which as a special case establishes criteria when two superdistributions may be multiplied. As an application of our framework, we study the singularities of distributional solutions of a supersymmetric field theory.
\end{abstract}



\maketitle


\section{Introduction and summary}
Supergeometry has its origins in theoretical physics, where it is used
as a refined model of spacetime that treats Bosonic and Fermionic
degrees of freedom on an equal footing. The basic concept
is that of a supermanifold, which loosely speaking is a manifold
with even (Bosonic) and odd (Fermionic) local coordinates. 
Quantum field theories on supermanifolds unify Bosonic and
Fermionic quantum fields in a single entity called a super quantum field. They
are very interesting from the perspective of a quantum field theorist
because of their improved renormalization behavior. Such special features
of supergeometric quantum field theories are collectively called 
non-renormalization theorems \cite{Grisaru,Seiberg}.
\sk

During the last decade, our mathematical understanding of perturbative
quantum field theory on Lorentzian manifolds has steadily improved, mainly due to
the development of perturbative algebraic quantum field theory (pAQFT), 
see e.g.\ Ref.~\onlinecite{Fredenhagen} for a recent review. In this framework, 
a key role is played by the class of Hadamard states, which are distinguished 
from a physical viewpoint since they share the same ultraviolet behavior of 
the Minkowski vacuum and they yield finite quantum fluctuations of all 
observables. From a mathematical perspective, they are defined in terms 
of a prescribed singular structure of the truncated two-point function
associated to the state \cite{Radzikowski}. Hence, in this respect, microlocal analysis serves 
as one of the main techniques used in pAQFT since 
its role is to analyze carefully the singularities of distributions like 
propagators and $n$-point functions. This proves essential not only for 
identifying Hadamard states but also for performing the perturbative 
construction and its renormalization.
\sk

The goal of this paper is to develop the foundations of microlocal analysis on supermanifolds.
Our work is based on and extends earlier investigations of Rempel and Schmitt \cite{Rempel}
on pseudodifferential operators on supermanifolds. As a new development, we introduce
a supergeometric generalization of the wavefront set, which is a suitable concept to encode 
polarization information about the singularities of distributions on supermanifolds.
See also Ref.~\onlinecite{Franco} for a first work in this direction, which however 
discards the polarized character of superdistributions.
Our super wavefront sets are motivated by the polarization sets of Dencker \cite{Dencker} 
for vector-valued distributions. However, they are constructed in such a way that they transform
in a natural way under supermanifold morphisms and not only vector bundle morphisms.
The techniques which we develop in this paper will be the basis to identify and to 
construct Hadamard states in the context of quantum field theories on supermanifolds. 
As mentioned before, these are characterized by a prescribed singular behavior of the associated, 
truncated two-point function and they are the building block for a covariant construction of 
Wick-polynomials. The latter are then used to introduce interaction terms within the perturbative framework.
Hence, the results of this paper are expected to play a major role in extending 
pAQFT to supergeometric quantum field theories \cite{HHS}, a longer term research goal that
we hope to achieve in future works. This would provide a rigorous framework to prove 
(and extend to curved supermanifolds) the non-renormalization theorems in Refs.~\onlinecite{Grisaru,Seiberg}.
\sk

The outline of the remainder of this paper is as follows:
In Section \ref{sec:prelim} we fix our notations and give a brief
review of some basic aspects of the theory of supermanifolds.
In Section \ref{sec:polarizationbundles} we assign
to each supermanifold $X = (\widetilde{X},\OO_X)$
a polarization bundle $\pi: \Pol X \to \TT^\ast\widetilde{X}$
over the cotangent bundle of the underlying smooth manifold $\widetilde{X}$;
this is a super vector bundle that encodes the local polarization information
of superfunctions and superdistributions on $X$. Our polarization bundle is a
special case of the general construction by Rempel and Schmitt in Section 8 of Ref.~\onlinecite{Rempel}:
It corresponds to a particular choice of what they call ``admissible tuple'', which is strongly motivated
by the fact that it enables us to detect ellipticity and hyperbolicity of the operators 
appearing in supergeometric field theories, see Examples \ref{ex:elliptic2} and \ref{ex:hyperbolic}.
In Section \ref{sec:superpseudos} we introduce super pseudodifferential operators
on supermanifolds, define their super principal symbols 
as bundle mappings between the polarization bundles, and develop
their calculus. The main definitions in this section are taken from Ref.~\onlinecite{Rempel} 
(see in particular Sections 7 and 8), which we however can present in a simplified 
form because of our particular choice of ``admissible tuple'' for the polarization bundles. 
We also present examples of super pseudodifferential operators which are relevant for physics, 
in particular the equation of motion operators (and their associated propagators) of
the supergeometric field theories studied in Ref.~\onlinecite{HHS}.
Crucially, as we have already indicated above, our concept of super principal symbols 
is able to detect ellipticity (or hyperbolicity) of these operators.
As our first genuinely new result, we introduce in Section \ref{sec:superwavefront} polarization sets for supermanifolds, 
motivated by Ref.~\onlinecite{Dencker}, and thereby define the super wavefront set
of a superdistribution. We analyze the transformation property of the 
super wavefront set under supermanifold morphisms and 
their compatibility with the action of super pseudodifferential operators.
In Section \ref{sec:pullback} we generalize to supermanifolds the ordinary pullback 
theorem for distributions on manifolds, see Theorem 8.2.4 in Ref.~\onlinecite{Hormander}. 
By including the polarization information of superdistributions (and their singularities),
this leads to a refinement of the ordinary pullback theorem. An important example 
is given by the super diagonal mapping, which provides criteria when two superdistributions
may be multiplied. As an application, we analyze in Section \ref{sec:singularities}
the singularities of distributional solutions to the equation of motion of the
$3\vert2$-dimensional Wess-Zumino model.


\section{\label{sec:prelim}Preliminaries}
We briefly recall some basic aspects of the theory of supermanifolds
which are frequently used in our work.
For a detailed introduction to this subject see, for example,
Refs.~\onlinecite{Carmeli,SuperBook} and also Section 2 in Ref.~\onlinecite{HHS} for a short summary.
\sk

A {\em superspace} is a pair $X = (\widetilde{X},\OO_X)$ consisting
of a topological space $\widetilde{X}$ (second-countable and Hausdorff)
and a sheaf of supercommutative superalgebras $\OO_X$ on $\widetilde{X}$, called the structure sheaf.
Explicitly, to each open $U\subseteq \widetilde{X}$ there is assigned a 
supercommutative superalgebra $\OO_X(U)$, called the sections of $\OO_X$ over $U$, 
and to each open $V\subseteq U\subseteq \widetilde{X}$ a superalgebra homomorphism 
$\mathrm{res}_{U,V} : \OO_X(U)\to \OO_X(V)$, called the restriction map.
The restriction maps satisfy the conditions
\begin{flalign}\label{eqn:presheaf}
\mathrm{res}_{U,U} = \id_{\OO_X(U)}~,\qquad 
\mathrm{res}_{V,W}\circ \mathrm{res}_{U,V} = \mathrm{res}_{U,W}~,
\end{flalign}
 for all open $W\subseteq V\subseteq U\subseteq \widetilde{X}$.
Moreover, given any open cover $\{U_\alpha\subseteq U\}$ of an open subset $U\subseteq \widetilde{X}$
and any matching family of local sections, i.e.\
\begin{flalign}
\left\{f_{\alpha}\in\OO_X(U_\alpha) :  \mathrm{res}_{U_\alpha, U_{\alpha\beta}}( f_{\alpha})
 =   \mathrm{res}_{U_\beta, U_{\alpha\beta}}( f_{\beta}) ~~\forall \alpha,\beta\right\}~,
\end{flalign}
where $U_{\alpha\beta} :=U_{\alpha}\cap U_{\beta}$ is the intersection,
there exists a unique section $f\in \OO_X(U)$ 
such that $f_{\alpha} = \mathrm{res}_{U,U_{\alpha}}(f)$.
Loosely speaking, this means that a family of local sections of $\OO_X$ which match in all overlaps 
can be glued to a unique global section and that any global section arises in that way.
\sk

The standard example of a superspace is $\bbR^{m\vert n} := (\bbR^m, C^\infty_{\bbR^m}\otimes \wedge^\bullet \bbR^n)$,
where $\wedge^\bullet \bbR^n$ denotes the Grassmann algebra with $n$ generators.
The sections over any open $U\subseteq \bbR^m$ are given by $C^\infty(U)\otimes \wedge^\bullet \bbR^n$.
Any element $f\in C^\infty(U)\otimes \wedge^\bullet \bbR^n$ has an expansion
\begin{flalign}
f = \sum_{I \in \bbZ_2^n}  f_I \,\theta^I := \sum_{(i_1,\dots,i_n) \in \bbZ_2^n}  f_{(i_1,\dots,i_n)} 
~{\theta^1}^{i_1}\cdots{\theta^n}^{i_n}~,
\end{flalign}
where $\bbZ_2^n := \{0,1\}^n$, $\{ \theta^a\in\bbR^n : a=1,\dots,n\}$ is the standard basis of $\bbR^n$
and $f_{I}\in C^\infty(U)$.
\sk

A morphism $\chi : X\to Y$ between two superspaces $X = (\widetilde{X},\OO_X)$ and $Y = (\widetilde{Y},\OO_Y)$
is a pair $(\widetilde{\chi} , \chi^\ast)$ consisting of a continuous map
$\widetilde{\chi} : \widetilde{X}\to \widetilde{Y}$ and a sheaf homomorphism
$\chi^\ast : \OO_Y \to \widetilde{\chi}_\ast \OO_X$, where  $\widetilde{\chi}_\ast \OO_X$ is the direct image sheaf.
Explicitly, to each open $U\subseteq \widetilde{Y}$ there is assigned
a superalgebra homomorphism $\chi^\ast_U : \OO_Y(U) \to \OO_X(\widetilde{\chi}^{-1}(U))$,
such that for all open $V\subseteq U\subseteq\widetilde{Y}$ the diagram
\begin{flalign}
\xymatrix{
\ar[d]_-{\mathrm{res}_{U,V}}\OO_Y(U) \ar[rr]^-{\chi^\ast_U} && \OO_X(\widetilde{\chi}^{-1}(U))\ar[d]^-{\mathrm{res}_{\widetilde{\chi}^{-1}(U),\widetilde{\chi}^{-1}(V)}}\\
\OO_Y(V) \ar[rr]_-{\chi^\ast_V}& &\OO_X(\widetilde{\chi}^{-1}(V))
}
\end{flalign}
commutes.
\sk

A {\em supermanifold} (of dimension $m\vert n$) is a superspace $X = (\widetilde{X},\OO_X)$
which is locally isomorphic to $\bbR^{m\vert n}$. More explicitly, this means that 
for any point $x\in \widetilde{X}$ there exists an open neighborhood $U\subseteq \widetilde{X}$ of $x$
such that $X\vert_{U} := (U,\OO_X\vert_U)$ is isomorphic as a superspace to
$W^{m\vert n} :=(W , C^\infty_W\otimes \wedge^\bullet\bbR^n)$, for
some open subset $W\subseteq \bbR^m$. We say that 
$\chi : X\to Y$ is a morphism between two supermanifolds $X = (\widetilde{X},\OO_X)$ and $Y = (\widetilde{Y},\OO_Y)$
if it is a superspace morphism.
\sk

Every supermanifold $X = (\widetilde{X},\OO_X)$ comes together with a filtration
\begin{flalign}\label{eqn:filtration}
\xymatrix{
\OO_X(U) & \ar[l] \JJ_X(U) & \ar[l] \JJ_X^2(U) & \ar[l] \cdots~~,
}
\end{flalign}
for any open $U\subseteq \widetilde{X}$, where
\begin{flalign}
\JJ_X(U) := \big\{f \in \OO_X(U) : f^N =0\, ,~\text{for some }N\in\bbN_0 \big\}\subseteq \OO_X(U)~
\end{flalign}
is the superideal of nilpotents
and $\JJ_X^k(U)$ is its $k$-th power, $k\geq 2$.
Locally, i.e.\ for sufficiently small $U\subseteq \widetilde{X}$, by definition there exists
an isomorphism $\OO_X(U) \simeq C^\infty(W)\otimes \wedge^\bullet \bbR^n$ of superalgebras 
for some open $W\subseteq \bbR^m$. Applying this isomorphism to the filtration \eqref{eqn:filtration} 
we obtain
\begin{flalign}\label{eqn:filtrationcoordinates}
\xymatrix{
C^\infty(W)\otimes \wedge^\bullet \bbR^n & \ar[l] C^\infty(W)\otimes \wedge^{\geq 1} \bbR^n & \ar[l] 
C^\infty(W)\otimes \wedge^{\geq 2} \bbR^n & \ar[l] \cdots~~,
}
\end{flalign}
which implies that locally $\JJ_X^k(U) =0$ for all $k>n$. Indeed, in this case
$C^\infty(W)\otimes \wedge^{\geq k} \bbR^n =0$.  Due to the sheaf condition the same 
statement holds globally, i.e.\ $\JJ_X^k(U) =0$ 
for all $k>n$ and $U\subseteq \widetilde{X}$ open.
\sk 

Let us also recall that to any $m\vert n$-dimensional 
supermanifold $X = (\widetilde{X},\OO_X)$ there is canonically 
assigned an $m$-dimensional manifold; it is specified by the topological space $\widetilde{X}$ together
with the structure sheaf $\OO_X/\JJ_X$. 
The underlying continuous map $\widetilde{\chi} : \widetilde{X} \to \widetilde{Y}$ of
any supermanifold morphism $\chi : X \to Y$ is smooth with respect to this manifold structure.
The supermanifold morphism $\iota_{\widetilde{X},X} :  (\widetilde{X},\OO_X/\JJ_X) \to (\widetilde{X},\OO_X)$,
given by $\widetilde{\iota}_{\widetilde{X},X} = \id_{\widetilde{X}}$
and the quotient mapping $\iota_{\widetilde{X},X}^\ast  : \OO_X \to \OO_X/\JJ_X$,
embeds the underlying smooth manifold into the supermanifold.


\section{\label{sec:polarizationbundles}Polarization bundles}
The space of superdistributions on a supermanifold $X$ is locally given by
$\DD^\prime(U)\otimes \wedge^\bullet \bbR^n$, where $U\subseteq \bbR^m$ is an open subset
and $\DD^\prime(U)$ denotes the space of distributions on $U$.
Hence, superdistributions locally carry polarization information in the Grassmann algebra
$\wedge^\bullet\bbR^n$. We now construct a bundle over the cotangent bundle 
$\TT^\ast \widetilde{X}$ of the underlying manifold $\widetilde{X}$, which describes 
the polarization information of superdistributions and their singularities.
Our construction in this section is a special case of the general construction 
by Rempel and Schmitt in Section 8 of Ref.~\onlinecite{Rempel}.
\sk

Let us start with the case where the supermanifold is a superdomain, 
i.e.\ ${U}^{m\vert n} := (U,C^\infty_{U}\otimes \wedge^{\bullet}\bbR^n)\subseteq \bbR^{m\vert n}$
for some open $U\subseteq \bbR^m$. In this case the polarization bundle is defined as
the trivial bundle
\begin{flalign}\label{eqn:trivialpolbundle}
\pi : \Pol U^{m\vert n} :=   \TT^\ast U \times \wedge^\bullet \bbC^n \longrightarrow \TT^\ast U~,~~(x,k,\lambda)\longmapsto (x,k)~,
\end{flalign}
where the fibers are the complexified Grassmann algebras
and $\TT^\ast U = U\times \bbR^m$ is the cotangent bundle over $U$.
\sk

Now consider a supermanifold  morphism $\chi : U^{m\vert n} \to V^{m^\prime\vert n^\prime}$ 
between two superdomains. The underlying smooth map $\widetilde{\chi} : U\to V$ induces
a fiber-wise pullback map 
$\TT^\ast\widetilde{\chi} :  \TT^\ast_{\widetilde{\chi}(x)} V \to \TT^\ast_{x}U$
of cotangent vectors, for any point $x\in U$.
Our goal is to construct a suitable fiber-wise 
map between the polarization bundles such that the diagram
\begin{flalign}\label{eqn:poldiagram}
\xymatrix{
\ar[d]_-{\pi} \Pol V^{m^\prime\vert n^\prime}\big\vert_{ \TT^\ast_{\widetilde{\chi}(x)} V}^{} \ar[rr]^-{\Pol \chi} && 
\Pol U^{m\vert n}\big\vert_{ \TT^\ast_{x} U}^{}\ar[d]^-{\pi}\\
\TT^\ast_{\widetilde{\chi}(x)} V \ar[rr]_-{\TT^\ast \widetilde{\chi}}&& \TT^\ast_{x} U
}
\end{flalign}
commutes, for any point $x\in U$. 
\sk

To approach this problem, we have to analyze in more detail the superalgebra homomorphism
$\chi^\ast_V : C^\infty(V) \otimes \wedge^\bullet\bbR^{n^\prime} \to C^\infty(U)\otimes 
\wedge^\bullet\bbR^{n}$. Using the (non-canonical!) $\bbZ$-gradings
\begin{flalign}
 C^\infty(V) \otimes \wedge^\bullet\bbR^{n^\prime} = \bigoplus_{i=0}^{n^\prime} 
  C^\infty(V) \otimes \wedge^i \bbR^{n^\prime} ~,\quad
   C^\infty(U)\otimes \wedge^\bullet\bbR^{n} = \bigoplus_{i=0}^n  C^\infty(U)\otimes 
\wedge^i\bbR^{n}~,
\end{flalign}
we decompose $\chi^\ast_V$ into components
\begin{flalign}
{{(\chi^\ast_V)}_{j}}^{i} : C^\infty(V) \otimes \wedge^i\bbR^{n^\prime} \longrightarrow C^\infty(U)\otimes 
\wedge^j\bbR^{n}~,
\end{flalign}
which are linear maps by construction. Notice that ${{(\chi^\ast_V)}_0}^0 = \widetilde{\chi}^\ast : 
C^\infty(V) \to C^\infty(U)$ is the pullback of functions along the underlying 
smooth map $\widetilde{\chi} : U\to V$.
We now show that the other components ${{(\chi^\ast_V)}_{j}}^{i}$ are
relative differential operators along $\widetilde{\chi}^\ast$.
Recall, e.g.\ from Theorem 4.1.11 in Ref.~\onlinecite{Carmeli}, that the superalgebra homomorphism $\chi^\ast_V$
is uniquely specified by its action on the supercoordinates $(y^{\mu^\prime},\zeta^{a^\prime})$
of $V^{m^\prime\vert n^\prime}$. We have that
\begin{flalign}\label{eqn:nilpotentorders}
\chi_V^\ast(y^{\mu^\prime}) - \widetilde{\chi}^\ast(y^{\mu^\prime}) \in \JJ_{U^{m\vert n}}^2(U)~,\qquad
\chi_V^\ast(\zeta^{a^\prime}) \in \JJ_{U^{m\vert n}}^1(U)~,
\end{flalign}
where $\JJ_{U^{m\vert n}}^k(U)$ is the filtration explained in \eqref{eqn:filtration}, see also
\eqref{eqn:filtrationcoordinates}. For a generic $f\in  C^\infty(V) \otimes \wedge^\bullet\bbR^{n^\prime}$,
we use the component expansion $f = \sum_{I\in\bbZ_2^{n^\prime}} f_I \,\zeta^I$ and obtain
\begin{flalign}
\chi_V^\ast(f) = \sum_{I\in\bbZ_2^{n^\prime}}\chi_V^\ast(f_I) \,\chi_V^\ast(\zeta^I)~.
\end{flalign}
Using the first property in
\eqref{eqn:nilpotentorders} and Taylor expansion in the odd coordinates,
we observe that
\begin{flalign}
\chi_V^\ast(f_I) = \widetilde{\chi}^\ast(f_I) +\sum_{l=1}^{\lfloor\frac{n}{2}\rfloor} \widetilde{\chi}^\ast\big(Q_{l} (f_I)\big)\,\lambda_{2l}~,
\end{flalign}
where $Q_{l}$ is a differential operator of order $l$ and $\lambda_{2l} \in\wedge^{2l} \bbR^{n}$.
Using also the second property in \eqref{eqn:nilpotentorders} and the fact that the odd coordinates $\theta^a$
on $U^{m\vert n}$ are nilpotent, we obtain
\begin{flalign}\label{eqn:pullbackcomponents}
{{(\chi^\ast_V)}_{j}}^{i} = \begin{cases}
\widetilde{\chi}^\ast\circ {(D^\chi)_{j}}^{i} &~,~~\text{if $j-i\geq 0$ even}~,\\
0 &~,~~\text{else}~.
\end{cases}
\end{flalign}
Here $ {(D^\chi)_{j}}^{i}$ are matrices of differential operators of order $\frac{j-i}{2}$.
In summary, we have shown that, for any supermanifold morphism
$\chi : U^{m\vert n} \to V^{m^\prime\vert n^\prime}$ between two
superdomains, the corresponding superalgebra homomorphism $\chi^\ast_V$ 
can be factorized uniquely as
\begin{flalign}
\chi^\ast_{V} = \widetilde{\chi}^\ast\circ D^\chi~,
\end{flalign}
where $D^\chi$ is a matrix of differential operators.
\sk

We now define the mapping $\Pol \chi$ in \eqref{eqn:poldiagram} component-wise by
\begin{flalign}\label{eqn:polmapping}
\nn {{(\Pol \chi)}_{j}}^i : \TT^\ast_{\widetilde{\chi}(x)} V \times\wedge^i\bbC^{n^\prime} &\longrightarrow
\TT^\ast_{x} U \times\wedge^j\bbC^{n}~,\\
\big(\widetilde{\chi}(x),k^\prime, \lambda^\prime\big) &\longmapsto
\begin{cases}
\big(x,\TT^\ast\widetilde{\chi}(k^\prime) , \sigma_{\frac{j-i}{2}}({(D^\chi)_{j}}^{i})(\widetilde{\chi}(x), k^\prime)\big(\lambda^\prime\big)\big)& ,~\text{if $j-i\geq 0$ even}~,\\
\big(x,\TT^\ast\widetilde{\chi}(k^\prime) ,0\big)&,~\text{else}~,
\end{cases}
\end{flalign}
where $\sigma_{l}$ denotes the principal symbol of a differential operator of order $l$.
\sk

Given now two supermanifold morphisms
$\chi : U^{m\vert n}\to V^{m^\prime\vert n^\prime}$
and $\chi^\prime : V^{m^\prime\vert n^\prime}\to W^{m^{\prime\prime}\vert n^{\prime\prime}}$,
we can form  the composition $\chi^\prime \circ \chi : U^{m\vert n}\to W^{m^{\prime\prime}\vert n^{\prime\prime}}$.
From $(\chi^\prime \circ \chi)^\ast_{W} = \chi^\ast_V\circ {\chi^\prime}^\ast_W$,
it follows that the components satisfy
\begin{flalign}
{{((\chi^\prime \circ \chi)^\ast_{W})}_{j}}^i = \sum_{h=0}^{n^\prime} {{( \chi^\ast_V)}_j}^h\circ  {{({\chi^\prime}^\ast_W)}_h}^i~
\end{flalign}
and hence
\begin{flalign}
\widetilde{\chi}^\ast\circ \widetilde{\chi^\prime}^\ast\circ {(D^{\chi^\prime\circ \chi})_{i+2l}}^{i} = \sum_{j=0}^l
\widetilde{\chi}^\ast\circ {(D^{\chi})_{i+2l}}^{i+2j}\circ \widetilde{\chi^\prime}^\ast\circ  {(D^{\chi^\prime})_{i+2j}}^{i}~,
\end{flalign}
for the non-vanishing components of ${{((\chi^\prime \circ \chi)^\ast_{W})}_{j}}^i$.
Combining this with \eqref{eqn:polmapping} and the multiplicativity of principal symbols, 
it is easy to check that the polarization mapping in \eqref{eqn:poldiagram}
is (contravariantly) compatible with compositions, i.e.\
\begin{flalign}\label{eqn:functoriality}
\Pol(\chi^\prime\circ \chi) = (\Pol\chi)\circ (\Pol\chi^\prime)~.
\end{flalign}
Moreover, by definition it is clear that $\Pol \id_{U^{m\vert n}} = \id_{\Pol U^{m\vert n}}$.
\sk

Because of this result, the concept of polarization bundle globalizes from superdomains to supermanifolds:
Let $X = (\widetilde{X}, \OO_X)$ be any $m\vert n$-dimensional supermanifold 
and choose an open cover $\{U_\alpha\subseteq \widetilde{X}\}$
and isomorphisms 
\begin{flalign}
\rho_\alpha : X\vert_{U_\alpha}^{} \longrightarrow {W_{\alpha}}^{m\vert n}\subseteq \bbR^{m\vert n}
\end{flalign}
 to superdomains, i.e.\ a superatlas. In all overlaps $U_{\alpha\beta} := U_\alpha\cap U_{\beta}$ this gives
rise to  transition supermanifold  morphisms
\begin{flalign}\label{eqn:transitionSMAN}
\chi_{\alpha\beta} := \rho_\beta \circ \rho_\alpha^{-1}: {W_{\alpha}}^{m\vert n}\vert_{\widetilde{\rho_\alpha}(U_{\alpha\beta})}^{} \longrightarrow
{W_{\beta}}^{m\vert n}\vert_{\widetilde{\rho_\beta}(U_{\alpha\beta})}^{}  ~,
\end{flalign}
which satisfy $\chi_{\alpha\alpha} =\id_{{W_{\alpha}}^{m\vert n}}$ for all $\alpha$
as well as the cocycle condition $\chi_{\beta\gamma}\circ \chi_{\alpha\beta} = \chi_{\alpha\gamma}$
on all triple overlaps $U_{\alpha\beta\gamma} := U_\alpha\cap U_{\beta}\cap U_{\gamma}$. 
In any superchart ${W_{\alpha}}^{m\vert n}$ 
we take the trivial polarization bundle $\Pol {W_{\alpha}}^{m\vert n}$  from \eqref{eqn:trivialpolbundle}.
The global polarization bundle $\Pol X$ on the supermanifold $X$ is then given by
gluing these local bundles via the transition functions
$g_{\alpha\beta} := \Pol\chi_{\beta\alpha}$;
the cocycle condition for the $g_{\alpha\beta}$ follows from \eqref{eqn:functoriality}.
It is important to stress that, even though the local polarization bundles \eqref{eqn:trivialpolbundle}
look like Grassmann algebra bundles, the transition functions $g_{\alpha\beta}$
in general {\em do not} preserve the product structure and the $\bbZ$-grading
on the fibers -- note the outer-diagonal terms in \eqref{eqn:polmapping}, which depend on $k$.
However, the coarser $\bbZ_2$-grading on the fibers of the local bundles
is preserved by the transition functions. Hence the polarization bundle 
$\pi : \Pol X \to \TT^\ast\widetilde{X}$ is a complex super vector bundle for any supermanifold $X = (\widetilde{X},\OO_X)$.


\section{\label{sec:superpseudos}Super pseudodifferential operators}
We introduce super pseudodifferential operators
on supermanifolds and define their super principal symbols. As in the case of a manifold,
the definition is local, and we first consider the case where the 
supermanifold is a superdomain ${U}^{m\vert n}\subseteq \bbR^{m\vert n}$.
The main definitions in this section are taken from Ref.~\onlinecite{Rempel} (see in particular
Sections 7 and 8). However, we will study the properties of super pseudodifferential operators
in more detail  and also provide interesting examples from supergeometric field theory.
\sk

A linear map
\begin{flalign}
A : C_{\mathrm{c}}^\infty(U)\otimes \wedge^\bullet\bbR^n\longrightarrow C^\infty(U)\otimes \wedge^\bullet\bbR^n
\end{flalign}
is called a {\em super pseudodifferential operator} on ${U}^{m\vert n}$ if all its components
${A_{j}}^i : C_{\mathrm{c}}^\infty(U)\otimes \wedge^i\bbR^n \to C^\infty(U)\otimes \wedge^j\bbR^n$
are (matrices of) pseudodifferential operators on $U\subseteq \bbR^m$.
In the following all pseudodifferential operators are implicitly assumed to be properly supported and classical, 
see e.g.\ Ref.~\onlinecite{Shubin} for the relevant definitions. Recall, in particular, that 
properly supported pseudodifferential operators map compactly supported functions
to compactly supported functions, hence they can be composed. 
The composition is again a properly supported pseudodifferential operator. 
Given any supermanifold {\em isomorphism} $\chi : {U}^{m\vert n} \to {V}^{m\vert n}$
and a super pseudodifferential operator $A$ on ${U}^{m\vert n}$,
consider the linear map
\begin{flalign}\label{eqn:pseudotransform}
{\chi_V^\ast}^{-1} \circ A\circ \chi_V^\ast: C_{\mathrm{c}}^\infty(V)\otimes \wedge^\bullet\bbR^n\longrightarrow C^\infty(V)\otimes \wedge^\bullet\bbR^n~.
\end{flalign}
It defines a super pseudodifferential operator on ${V}^{m\vert n}$
because the components of  $\chi_V^\ast$ and its inverse are both (matrices of) 
relative differential operators, cf.\ \eqref{eqn:pullbackcomponents}. 
\begin{defi}\label{defi:pseudosymbol}
We say that a super pseudodifferential operator $A$ on ${U}^{m\vert n}$ 
is of order $l$ if its components ${A_{j}}^i$ are (matrices of) pseudodifferential operators on $U$ 
of order $\frac{j-i}{2} +l$, i.e.,
\begin{flalign}\label{sPSIDO}
s\Psi\mathrm{DO}^{l}({U}^{m\vert n}) := \left\{A :C_{\mathrm{c}}^\infty(U)\otimes \wedge^\bullet\bbR^n\to 
C^\infty(U)\otimes \wedge^\bullet\bbR^n :   {A_{j}}^i\in \Psi\mathrm{DO}^{\frac{j-i}{2}+l}(U) \right\}~.
\end{flalign}
The super principal symbol of $A\in s\Psi\mathrm{DO}^l(U^{m\vert n})$ is the super vector bundle map
\begin{flalign}
\sigma_l(A) : \Pol U^{m\vert n} \longrightarrow \Pol U^{m\vert n}~
\end{flalign}
with components given by
\begin{flalign}
\nn {{(\sigma_l(A))}_j}^i : \TT^\ast U \times \wedge^i \bbC^n &\longrightarrow \TT^\ast U\times \wedge^j\bbC^n~,\\
(x,k,\lambda) &\longmapsto \big(x,k,\sigma_{\frac{j-i}{2}+l}({A_{j}}^i)(x,k) \big(\lambda\big)\big)~,
\end{flalign}
where $\sigma_{\frac{j-i}{2}+l}({A_j}^i)$ is the ordinary principal symbol of order $\frac{j-i}{2}+l$ of ${A_j}^i$.
\end{defi}
\begin{ex}\label{ex:supermorphaspseudo}
Let $\chi : U^{m\vert n} \to V^{m\vert n}$ be a supermanifold isomorphism
between two superdomains, and consider the unique factorization $\chi_V^\ast = \widetilde{\chi}^\ast\circ D^{\chi}$
given in \eqref{eqn:pullbackcomponents}. Then $D^\chi$ is a super pseudodifferential operator
of order $0$, i.e.\ $D^\chi\in s\Psi\mathrm{DO}^0(V^{m\vert n})$.
In the case where $U=V$ and $\widetilde{\chi} = \id_U$, the super principal symbol
of $D^\chi$ is the polarization mapping \eqref{eqn:polmapping}, i.e.\ $\sigma_0(D^\chi) = \Pol \chi$.
\end{ex}

We collect some useful properties of super pseudodifferential operators
and their super principal symbols. The proofs of these statements follow easily from our definitions
and are omitted.
\begin{lem}\label{lem:pseudoproperties}
Let $A\in s\Psi\mathrm{DO}^{l}({U}^{m\vert n})$ and $B\in s\Psi\mathrm{DO}^{l^\prime}({U}^{m\vert n})$.
Then the following statements hold true:
\begin{itemize}
\item[a)] $B\circ A \in s\Psi\mathrm{DO}^{l+l^\prime}({U}^{m\vert n})$.
\item[b)] If $\chi : {U}^{m\vert n} \to {V}^{m\vert n}$ is a supermanifold isomorphism,
then ${\chi_V^{\ast}}^{-1} \circ A\circ \chi_V^\ast\in s\Psi\mathrm{DO}^{l}({V}^{m\vert n})$.
\end{itemize}
\end{lem}
\begin{lem}\label{lem:symbolproperties}
Let $A\in s\Psi\mathrm{DO}^{l}({U}^{m\vert n})$ and $B\in s\Psi\mathrm{DO}^{l^\prime}({U}^{m\vert n})$.
Then the following statements hold true:
\begin{itemize}
\item[a)] $\sigma_{l+l^\prime} (B \circ A) = \sigma_{l^\prime}(B)\circ \sigma_{l}(A)$.
\item[b)]  If $\chi : {U}^{m\vert n} \to {V}^{m\vert n}$ is a supermanifold isomorphism,
then 
\begin{flalign}
\sigma_{l}\big({\chi_V^{\ast}}^{-1} \circ A\circ \chi_V^\ast\big) = 
(\Pol \chi^{-1})\circ \sigma_{l}(A)\circ (\Pol \chi)~.
\end{flalign}
\end{itemize}
\end{lem}

Super pseudodifferential operators and their super principal symbols are
easily globalized to supermanifolds by slightly adapting the globalization
procedure for the pseudodifferential operators on manifolds, see e.g.\ Chapter I, Section 5 in Ref.~\onlinecite{Treves}. 
Let  $X = (\widetilde{X}, \OO_X)$ be an $m\vert n$-dimensional supermanifold and 
$\mathcal{O}_{X,\mathrm{c}}(\widetilde{X})$ the space of
compactly supported global sections of the structure sheaf.
Consider a maximal superatlas $\rho_{\alpha} : X\vert_{U_\alpha}^{} \to
{W_{\alpha}}^{m\vert n}$. 
A super pseudodifferential operator $A\in s\Psi\mathrm{DO}^l(X)$ of order $l$ on $X$
is a continuous linear map $A : \mathcal{O}_{X,\mathrm{c}}(\widetilde{X})\to\mathcal{O}_{X}(\widetilde{X})$ 
such that, for every superchart ${W_{\alpha}}^{m\vert n}$, the linear map
$A_{\alpha}$ defined by the diagram
\begin{flalign}\label{localPseudo}
\xymatrix{
\ar[d]_-{\rho_{\alpha}^\ast}C^\infty_{\mathrm{c}}(W_{\alpha})\otimes \wedge^\bullet \bbR^n \ar[rrr]^-{A_{\alpha}}&&& C
^\infty(W_{\alpha})\otimes \wedge^\bullet \bbR^n\\
\ar[r]_-{\mathrm{ext}_{\widetilde{X},U_{\alpha}}}\OO_{X,\mathrm{c}}(U_{\alpha})&
\OO_{X,\mathrm{c}}(\widetilde{X})\ar[r]_-{A} & 
 \OO_{X}(\widetilde{X})\ar[r]_-{\res_{U_{\alpha},\widetilde{X}}} & \OO_{X}(U_{\alpha}) \ar[u]_-{{\rho_{\alpha}^\ast}^{-1}}
}
\end{flalign}
is an element in $s\Psi\mathrm{DO}^l({W_{\alpha}}^{m\vert n})$. Here $\mathrm{ext}$ denotes
the extension (by zero) maps for compactly supported sections.
To each $A\in s\Psi\mathrm{DO}^l(X)$ we associate a super principal symbol, 
which is a super vector bundle morphism
\begin{flalign}\label{eqn:globalsymbol}
\sigma_{l}(A) : \Pol X\longrightarrow \Pol X~.
\end{flalign}
Explicitly, the super principal symbol $\sigma_{l}(A)$ is constructed by gluing together the collection of all local super principal 
symbols $\sigma_l(A_\alpha)$ of the operators $A_\alpha$ in \eqref{localPseudo}. This is consistent 
on account of Lemma \ref{lem:symbolproperties} b).
\sk

To study the singularities of distributions, the notion of ellipticity is crucial.
\begin{defi}\label{def:ellipticpseudo}
We say that a super pseudodifferential operator $E\in s\Psi \mathrm{DO}^l(X)$ 
is {\em elliptic} if the super principal symbol $\sigma_l(E)$ is invertible on $\TT^\ast\widetilde{X}\setminus {\bf 0}$.
\end{defi}

Many properties of elliptic pseudodifferential operators on ordinary manifolds 
are still valid in our framework. In particular, we obtain
\begin{lem}\label{lem:inverseOP}
Let  $E\in s\Psi \mathrm{DO}^l(X)$ be an elliptic super pseudodifferential operator. 
Then there exists a super pseudodifferential operator $F \in s\Psi \mathrm{DO}^{-l}(X)$ such that
\begin{flalign}
E\circ  F -  \id  \in s\Psi \mathrm{DO}^{-\infty}(X) \quad \text{and} \quad  F \circ E - \id  \in s\Psi \mathrm{DO}^{-\infty}(X)~,
\end{flalign}
where $s\Psi \mathrm{DO}^{-\infty}(X) := \bigcap_{l\in\bbR} s\Psi \mathrm{DO}^{l}(X)$. $F$ is called a parametrix for $E$.
\end{lem}
\begin{proof}
The proof is as in the case of ordinary manifolds, see e.g.\ Theorem 5.1 in Ref.~\onlinecite{Shubin}.
\end{proof}

We shall now give examples of super differential and super pseudodifferential operators $A\in s\Psi\mathrm{DO}^l(X)$
which have their origin in supersymmetric field theory.
\begin{ex}\label{ex:elliptic2}
Let $X = \bbR^{1\vert 1}$ be the superline. The dynamics of a superparticle on $X$ 
is governed 
by a super differential operator, which in global supercoordinates $(t,\theta)$ on  $\bbR^{1\vert 1}$ reads as
\begin{flalign}
\nn P : C^\infty(\bbR)\otimes \wedge^\bullet\bbR &\longrightarrow C^\infty(\bbR)\otimes \wedge^\bullet\bbR~,\\
f= f_0 + f_1\,\theta &\longmapsto \partial_t f_1 + \partial_t^2 f_0 \,\theta~,
\end{flalign}
cf.\ Section 8.1 in Ref.~\onlinecite{HHS}. In our component notation, the operator $P$ is given by
\begin{flalign}
P = \begin{pmatrix}
0 & \partial_t \\
\partial_t^2 & 0
\end{pmatrix}~.
\end{flalign}
Notice that $P\in s\Psi\mathrm{DO}^{\frac{3}{2}}(\bbR^{1\vert 1})$. 
Its super principal symbol
\begin{flalign}
\sigma_{\frac{3}{2}}(P)(t,k) = \begin{pmatrix}
0 & \ii k \\
-k^2 & 0
\end{pmatrix}
\end{flalign}
is invertible for all $(t,k) \in \TT^\ast \bbR \setminus {\bf 0}$, hence $P$ is elliptic. Specifically, the inverse is
	\begin{flalign}
		\sigma_{-\frac{3}{2}}(F)(t,k):=\sigma_{\frac{3}{2}}(P)^{-1}(t,k) = \begin{pmatrix}
			0 &  -\frac{1}{k^2} \\
			-\frac{\ii}{k} & 0
		\end{pmatrix}~.
	\end{flalign}
	 In this case a parametrix $F$ of $P$ from Lemma \ref{lem:inverseOP} is explicitly given by
	 the integral kernel
	\begin{flalign}
		F (t,t^\prime)=\frac{1}{2}\, \begin{pmatrix}
			0 &   (t-t^\prime)\,\text{sign}(t-t^\prime)\\
			\text{sign}(t-t^\prime) & 0
		\end{pmatrix}~.
	\end{flalign} 
\end{ex}

\begin{ex}\label{ex:hyperbolic}
Let us consider $X = (M,C^\infty_M\otimes\wedge^\bullet\bbR^2)$, where
$M$ is a smooth 3-dimensional Lorentzian manifold.
The equation of motion operator  $P : \OO_X(M)\to \OO_X(M)$
of the $3\vert 2$-dimensional Wess-Zumino model on $X$ is then given in component notation by
\begin{flalign}\label{eqn:EOM3D}
P= \begin{pmatrix}
m & 0 & -1\\
0 & \ii\!\!\not \!\nabla +m & 0\\
\square & 0 & m
\end{pmatrix}~,
\end{flalign}
cf.\ Section 8.2 in Ref.~\onlinecite{HHS}. 
Here $\ii\!\!\not \!\nabla$ is the Dirac operator (on $M$), $\square$ is the d'Alembert operator on $M$
and $m\geq 0$ is a mass term. Notice further that the component notation
in \eqref{eqn:EOM3D} is in block-matrix form, 
because $\wedge^1\bbR^2 \simeq \bbR^2$ is two-dimensional; in particular, the Dirac operator
is a $2\times 2$-matrix of differential operators.
The operator $P\in s\Psi\mathrm{DO}^{1}(X)$ is of order $1$, and in local 
coordinates $x^\mu$ and $k_{\mu}$ on $\TT^\ast M$ its super principal symbol is given by
\begin{flalign}\label{eqn:EOM3Dsymbol}
\sigma_{1}(P)(x,k) = 
\begin{pmatrix}
0 & 0 & -1\\
0 & -\gamma^\mu(x)\,k_{\mu} & 0\\
- k_{\mu}k_{\nu} \,g^{\mu\nu}(x) & 0 & 0
\end{pmatrix}~.
\end{flalign} 
Using the Clifford algebra relations $\{\gamma^\mu,\gamma^\nu\} = 2\,g^{\mu\nu}$ for the gamma-matrices,
it is easy to check that $\sigma_{1}(P)(x,k)$ is invertible for all $(x,k)\in \TT^\ast M\setminus {\bf 0}$ 
which are not light-like (i.e.\ $k_\mu k_\nu g^{\mu\nu}(x) \neq 0$). More explicitly, we have
\begin{flalign}
\sigma_{1}(P)(x,k)^{-1} =  
\begin{pmatrix}
0 & 0 & -\frac{1}{ k_{\mu}k_{\nu} \,g^{\mu\nu}(x)}\\
0 & -\frac{\gamma^\mu(x)\,k_{\mu}}{{ k_{\mu}k_{\nu} \,g^{\mu\nu}(x)}} & 0\\
- 1 & 0 & 0
\end{pmatrix}.
\end{flalign}
Because $\sigma_{1}(P)(x,k)$ is invertible for non-light-like $(x,k)\in \TT^\ast M\setminus {\bf 0}$,
we call $P$ hyperbolic.
\end{ex}
\begin{rem}
Our definition of orders and super principal symbols for super pseudodifferential operators on supermanifolds
is well suited for the examples of super (pseudo-)differential operators arising in supersymmetric field theory.
This is a consequence of our definition of the polarization bundle $\pi :\Pol X \to \TT^\ast\widetilde{X}$
and in particular of the assignment of the polarization mapping defined in \eqref{eqn:polmapping}.
Rempel and Schmitt \cite{Rempel} consider also more general polarization bundles
(defined via polarization mappings different from \eqref{eqn:polmapping}), which are classified
by what they call admissible tuples. It is important to stress that
all other polarization bundles in Ref.~\onlinecite{Rempel} 
lead to an assignment of orders and super principal symbols for
super pseudodifferential operators on $X$ which is not able to detect ellipticity and hyperbolicity in
our examples above. This provides us  with a motivation for our choice of polarization bundle
given in \eqref{eqn:polmapping}.
\end{rem}


\section{\label{sec:superwavefront}Super wavefront sets}
We start with the case where the 
supermanifold is a superdomain 
${U}^{m\vert n}\subseteq \bbR^{m\vert n}$.
Then the space of superdistributions $\DD^\prime(U)\otimes \wedge^\bullet \bbR^n$ 
is the dual of $C_{\mathrm{c}}^\infty(U)\otimes \wedge^\bullet \bbR^n$, and both 
$C_{\mathrm{c}}^\infty(U)\otimes \wedge^\bullet \bbR^n$ and $C^\infty(U)\otimes \wedge^\bullet \bbR^n$ are dense sub-spaces.
We say that a superdistribution $u\in \DD^\prime(U)\otimes \wedge^\bullet \bbR^n $
is smooth if it is an element of $C^\infty(U)\otimes \wedge^\bullet \bbR^n$.
Crucially, by duality, any (properly supported) super pseudodifferential operator $A$ on ${U}^{m\vert n}$
admits a continuous extension to superdistributions,\
$A : \DD^\prime(U)\otimes \wedge^\bullet \bbR^n\to \DD^\prime(U)\otimes \wedge^\bullet \bbR^n$.
Global superdistributions on a supermanifold $X$ are obtained by gluing 
local superdistributions in a superatlas, 
via the transition morphisms $\chi_{\alpha\beta}$ given in \eqref{eqn:transitionSMAN}.
\sk

We define the super wavefront set of a superdistribution on $X$ motivated by the approach of Dencker \cite{Dencker}
for vector-valued distributions. 
The starting point is the polarization bundle $\pi: \Pol X\to\TT^\ast \widetilde X$ introduced in Section 
\ref{sec:polarizationbundles}. We denote by 
\begin{flalign}
\pi: \widehat{\mathcal{P}}^\ast X := \pi^{-1}\big(\TT^\ast \widetilde X\setminus {\bf 0}\big) \longrightarrow\TT^\ast \widetilde X \setminus {\bf 0}
\end{flalign}
the restriction of the polarization bundle to the cotangent bundle with the zero-section removed.
\begin{defi}\label{sWF}
The super wavefront set (of order $l$) of a superdistribution 
$u\in \DD^\prime(U)\otimes \wedge^\bullet \bbR^n $ 
is defined as the intersection
\begin{flalign}
s\mathrm{WF}^l(u) := \bigcap_{\mycom{A\in s\Psi\mathrm{DO}^l(U^{m\vert n})}{\text{s.t.\ }Au\text{ smooth}}}\left\{(x,k,\lambda)\in \widehat{\mathcal{P}}^\ast  U^{m\vert n} \,:\, \sigma_{l}(A)(x,k)\big(\lambda\big) =0  \right\}\subseteq \widehat{\mathcal{P}}^\ast U^{m\vert n}~.
\end{flalign}
\end{defi}

We collect some important properties of the super wavefront sets defined above.
\begin{propo}\label{propo:WFproperties}
For any $u\in \DD^\prime(U)\otimes \wedge^\bullet \bbR^n $, the following properties hold true:
\begin{itemize}
\item[a)] $s\mathrm{WF}^l(u) = s\mathrm{WF}^{l^\prime}(u)$ for all $l,l^\prime$.
\item[b)] For $u = \sum_{I\in\bbZ_2^n} u_I\,\theta^I \in \DD^\prime(U)\otimes \wedge^\bullet \bbR^n $,  
\begin{flalign}\label{projWF}
\pi\left(s\mathrm{WF}^l(u)\setminus \big( ( \TT^\ast U\setminus{\bf 0}) \times \{0\}\big)\right) 
= \bigcup_{I\in\bbZ_2^n} \mathrm{WF}(u_I)~,
\end{flalign}
where $\pi : \widehat{\mathcal{P}}^\ast U^{m\vert n} \to \TT^\ast U\setminus {\bf 0}$ 
is the projection \eqref{eqn:trivialpolbundle} and $ \mathrm{WF}(u_I)\subseteq \TT^\ast U\setminus {\bf 0}$ 
denotes the ordinary wavefront set of $u_I\in \DD^\prime(U)$. 
\end{itemize}
\end{propo}
\begin{proof}
To show item a), take any $(x,k,\lambda)\not\in s\mathrm{WF}^l(u)$. By assumption 
there exists $A\in s\Psi\mathrm{DO}^l(U^{m\vert n})$ such that 
$Au$ smooth and $\sigma_l(A)(x,k)\big(\lambda\big)\neq 0$.
Composing this $A$ with any elliptic super pseudodifferential operator
$E \in s\Psi\mathrm{DO}^{l^\prime-l}(U^{m\vert n})$ of order $l^\prime-l$,
we obtain $E\circ A\in s\Psi\mathrm{DO}^{l^\prime}(U^{m\vert n})$ such that
$EAu$ smooth and $\sigma_{l^\prime}(E\circ A)(x,k)\big(\lambda\big) =
\sigma_{l^\prime- l}(E)(x,k)\big(\sigma_{l}(A)(x,k) \big(\lambda\big)\big)\neq 0$.
Hence, $(x,k,\lambda)\not\in s\mathrm{WF}^{l^\prime}(u)$, which completes the proof.
\sk

Item b): We prove the inclusion ``$\subseteq$'' by contradiction. 
Suppose that there exists $(x,k,\lambda)\in 
s\mathrm{WF}^l(u)\setminus ( ( \TT^\ast U\setminus{\bf 0}) \times \{0\})$
such that $(x,k)\not\in \bigcup_{I \in \mathbb{Z}_2^n} \mathrm{WF}(u_I)$.
The latter condition implies that, for each $I\in \mathbb{Z}_2^n$,
there exists $A_I\in \Psi\mathrm{DO}^l(U)$ such that $A_I u_I$ is smooth and
$\sigma_l(A_I)(x,k)\neq 0$. We define $A\in s\Psi\mathrm{DO}^l(U^{m\vert n})$
by placing the $A_I$ in their corresponding diagonal entry of the matrix 
and setting all other entries to zero.
By construction, we have that $A u$ is smooth and 
that the super principal symbol $\sigma_l(A)(x,k)$ is invertible. 
This implies that $\lambda=0$ and leads to a contradiction.
\sk

We prove the inclusion ``$\supseteq$'' by contradiction. 
Suppose that there exists an element $(x,k)\in \bigcup_{I \in \mathbb{Z}_2^n} \mathrm{WF}(u_I)$
such that $(x,k,\lambda)\not \in s\mathrm{WF}^l(u)\setminus ( ( \TT^\ast U\setminus{\bf 0}) \times \{0\})$,
for any $\lambda\neq 0$. Then there exists $A\in s\Psi\mathrm{DO}^l(U^{m\vert n})$ such that
$Au $ is smooth and $\sigma_l(A)(x,k)$ is invertible at $(x,k)$. Thus, by a straightforward refinement 
of Lemma \ref{lem:inverseOP}, as in Proposition 6.9 in Ref.~\onlinecite{Treves} we construct a 
microlocal parametrix $F\in s\Psi\mathrm{DO}^{-l}(U^{m\vert n})$. From the existence of this microlocal 
parametrix $F$ we conclude that all components $u_I$ of $u$ are smooth at $(x,k)$. Hence 
$(x,k)\notin \bigcup_{I \in \mathbb{Z}_2^n} \mathrm{WF}(u_I)$, which is a contradiction.
\end{proof}
\begin{rem}
On account of item a) of the previous lemma, we drop the label $l$ and denote the super wavefront set by $s\mathrm{WF}(u)$.
\end{rem}

\begin{cor}\label{cor:minsWF}
$u\in \DD^\prime(U)\otimes \wedge^\bullet \bbR^n$ is smooth if and only if
$s\mathrm{WF}(u) = ( \TT^\ast U\setminus {\bf 0}) \times \{0\}$.
\end{cor}
\begin{proof}
The statement is a special instance of \eqref{projWF}.
\end{proof}

\begin{ex}\label{ex:sWF}
Let us consider the superdomain $U^{m\vert 2}$ and the superdistribution
\begin{flalign}
u = v + v\,\theta^1\theta^2 = \begin{pmatrix}
v \\ 0 \\ v
\end{pmatrix}~,
\end{flalign} 
where $v\in \DD^\prime(\bbR^{m})$ is an ordinary distribution and
$0$ denotes the zero vector in $\wedge^1\bbR^2\simeq \bbR^2$ 
according to our block-matrix component notation.
Then the super pseudodifferential operator
\begin{flalign}
A = \begin{pmatrix}
0 & 0 & 0\\
0 & 0 & 0\\
-1 & 0 & 1
\end{pmatrix}
\end{flalign}
is of order $0$ and annihilates $u$. In particular, $Au=0$ is smooth.
The super principal symbol of order $0$ of $A$ reads as
\begin{flalign}
\sigma_0(A)(x,k) = \begin{pmatrix}
0 & 0 & 0\\
0 & 0 & 0\\
0 & 0 & 1
\end{pmatrix}~,
\end{flalign}
for any $(x,k) \in \TT^\ast U$. Hence all polarization vectors in the super wavefront set 
$s\mathrm{WF}(u)$ in Definition \eqref{sWF} have necessarily a vanishing third component 
(i.e.\ highest component in the $\theta$-expansion). 
Explicitly,
\begin{flalign}
s\mathrm{WF}(u) \subseteq \big(\TT^\ast U\setminus {\bf 0}\big) \times \big\{\lambda \in \wedge^\bullet \bbC^2 \,:\, \lambda_{(1,1)}=0 \big\}~.
\end{flalign}
Loosely speaking, this shows that our notion of super wavefront sets both picks 
out the leading singularities to determine the polarization and assigns a higher weight 
to the components of a superdistribution with a lower number of $\theta$-powers. 
Notice that this is a direct consequence of our definition of orders and super principal symbols for super pseudodifferential
operators in Definition \ref{defi:pseudosymbol}. Hence this feature generalizes to superdomains in
higher odd-dimensions $U^{m\vert n}$.
\end{ex}

The super wavefront set of a superdistribution behaves well with respect 
to the action of super pseudodifferential operators.
\begin{propo}\label{propo:sWFpseudoaction}
Let $u\in \DD^\prime(U)\otimes \wedge^\bullet \bbR^n $ and $A\in s\Psi\mathrm{DO}^l(U^{m\vert n})$.
Then
\begin{flalign}\label{eqn:sWFpseudoaction}
s\mathrm{WF}(Au) \supseteq \sigma_l(A)\big(s\mathrm{WF}(u)\big) := \left\{\big(x,k,\sigma_l(A)(x,k)\big(\lambda\big)\big) : (x,k,\lambda)\in s\mathrm{WF}(u) \right\}~,
\end{flalign}
where the equality holds true whenever $A$ is elliptic.
\end{propo}
\begin{proof}
Let  $(x,k,\lambda) \in s\mathrm{WF}(u)$ and 
$B\in  s\Psi \mathrm{DO}^{l'}(U^{m\vert n})$
be such that $BAu$ is smooth. 
By hypothesis, we have that $\sigma_{l+l^\prime} (B\circ A)(x,k)\big(\lambda\big)=0$,
and hence $\sigma_{l^\prime}(B)(x,k)\big( \sigma_l(A)(x,k)\big(\lambda\big)\big) =0$.
As $B$ was arbitrary (as long as $BAu$ is smooth),
this implies that $\big(x,k,\sigma_l(A)(x,k)\big(\lambda\big)\big) \in  s\mathrm{WF}(Au)$.
\sk

If $A$ is elliptic, we use Lemma \ref{lem:inverseOP} to obtain
an elliptic $F\in s\Psi\mathrm{DO}^{-l}(U^{m\vert n})$, 
such that both $A\circ F-\id$ and $F\circ A-\id$ 
lie in $s\Psi\mathrm{DO}^{-\infty}(U^{m\vert n})$. 
Equality in \eqref{eqn:sWFpseudoaction} 
is then shown by replacing the role of $u$ with $Au$ and that of $A$ with $F$.
\end{proof}
\begin{rem}
More generally, equality in \eqref{eqn:sWFpseudoaction} 
holds true microlocally above any point $(x,k)\in \TT^\ast U\setminus{\bf 0}$ 
where $\sigma_l(A)$ is invertible.
\end{rem}

Given any supermanifold isomorphism $\chi : U^{m\vert n}\to V^{m\vert n}$,
the fibre-wise polarization mapping from \eqref{eqn:polmapping}
defines a super vector bundle isomorphism
\begin{flalign}
\xymatrix{
\ar[d]_{\pi} \Pol  V^{m\vert n} \ar[rr]^-{\Pol \chi} &&  \Pol U^{m\vert n}\ar[d]^-{\pi}\\
\ar[d]_{\pi_{\TT^\ast}}\TT^\ast V \ar[rr]^-{\TT^\ast \widetilde{\chi}} && \TT^\ast U \ar[d]^{\pi_{\TT^\ast}}\\
V \ar[rr]_-{\widetilde{\chi}^{-1}} && U
}
\end{flalign}
We now show that the super wavefront sets transform well under supermanifold isomorphisms.
\begin{propo}\label{propo:trafosWF}
Let $\chi : U^{m\vert n}\to V^{m\vert n}$ be a supermanifold isomorphism
and $u \in \DD^\prime (V)\otimes \wedge^\bullet\bbR^n$ a superdistribution. 
Denote by $\chi^\ast_V(u)\in \DD^\prime (U)\otimes \wedge^\bullet\bbR^n $
the pullback of $u$ along $\chi$. Then
\begin{flalign}
s\mathrm{WF}(\chi^\ast_V(u)) = \Pol \chi \big(s\mathrm{WF}(u)\big)~.
\end{flalign}
\end{propo}
\begin{proof}
This is a direct consequence of Lemma \ref{lem:symbolproperties} b).
\end{proof}

This transformation property of the super wavefront set 
under the action of all supermanifold isomorphisms allows us to globalize
super wavefront sets from superdomains to supermanifolds:
Let $u$ be a superdistribution on a supermanifold $X = (\widetilde{X},\OO_X)$.
We use a superatlas $\rho_{\alpha} :X\vert_{U_{\alpha}}^{} \to {W_{\alpha}}^{m \vert n}$ 
and describe $u$ in terms of a family of local superdistributions 
$u_{\alpha}\in \DD^\prime(W_{\alpha})\otimes\wedge^\bullet\bbR^n$,
which satisfy the gluing conditions
\begin{flalign}
\res_{W_{\beta}, \widetilde{\rho_\beta}(U_{\alpha\beta})}\big(u_{\beta}\big) 
= {\chi_{\beta\alpha}}^\ast \big(\res_{W_{\alpha}, \widetilde{\rho_\alpha}(U_{\alpha\beta})} (u_{\alpha})\big)
\end{flalign}
on all overlaps $U_{\alpha\beta}$.
Here $\chi_{\beta\alpha}$ are the transition  supermanifold  morphisms.
The super wavefront set of $u$ 
is then obtained by gluing all subsets $s\mathrm{WF}(u_{\alpha}) \subseteq \Pol {W_{\alpha}}^{m\vert n}$
via the transition functions $g_{\alpha\beta} = \Pol \chi_{\beta\alpha}$ of the polarization bundle.
Proposition \ref{propo:trafosWF} guarantees that this construction 
defines a global super wavefront set $s\mathrm{WF}(u) \subseteq \Pol X$.


\section{\label{sec:pullback}Pullback and multiplication theorems}
Given a {\em generic} supermanifold morphism $\chi : X\to Y$,
 we {\em cannot} pull back a generic superdistribution $u$ on $Y$ to 
a superdistribution on $X$. However, depending on the explicit form of $\chi$,
certain superdistributions $u$ on $Y$ may admit a (unique) pullback to $X$. 
It is the goal of this section to develop a suitable criterion to select a class of
superdistributions which admit a pullback.
\sk

Before we start with supergeometric considerations, 
let us briefly recall the solution to the above problem in ordinary geometry, 
see e.g.\ Ref.~\onlinecite{Hormander}: Consider 
a smooth map $\widetilde{\chi} : U\to V$ between two open domains $U\subseteq \bbR^m$ and $V\subseteq \bbR^{m^\prime}$. 
The normal set of $\widetilde{\chi}$ is the subset of $\TT^\ast V$ given by
\begin{flalign}
N_{\widetilde{\chi}} := \Big\{ \big(\widetilde{\chi}(x),k^\prime\big) \in \TT^\ast V \, :\, x\in U\, ,~ \TT^\ast \widetilde{\chi} (k^\prime)=0  \Big\}~.
\end{flalign}
It was shown in Theorem 8.2.4 in Ref.~\onlinecite{Hormander} that the pullback map
$\widetilde{\chi}^\ast : C^\infty(V)\to C^\infty(U)$ admits a unique continuous extension 
to those distributions $u\in \DD^\prime(V)$ for 
which $\mathrm{WF}(u)\cap N_{\widetilde{\chi}}=\emptyset$ holds true.
\sk

Let us now consider a supermanifold morphism $\chi : U^{m\vert n}\to V^{m^\prime\vert n^\prime}$
between two superdomains. The case of a generic supermanifold morphism $\chi : X\to Y$ between two supermanifolds
follows from this by localizing $\chi$ in suitable superatlases of $X$ and $Y$.
Recalling that $\chi_V^\ast$ admits a unique factorization \eqref{eqn:pullbackcomponents}
into a matrix of differential operators $D^\chi$ and the component-wise pullback $\widetilde{\chi}^\ast$ along
the underlying smooth map, we analyze the pullback of superdistributions in two steps:
Given any superdistribution $u\in \DD^\prime (V)\otimes \wedge^\bullet \bbR^{n^\prime}$ on $V^{m^\prime\vert n^\prime}$,
the first step is to act with the differential operator $D^\chi$ on $u$, which 
is {\em always} well-defined and results in an auxiliary superdistribution
\begin{flalign}
D^\chi u \in \DD^\prime (V)\otimes \wedge^\bullet \bbR^{n}~,
\end{flalign}
where the components are now in the Grassmann algebra $\wedge^\bullet \bbR^{n}$ with $n$ generators.
In the second step, we would like to pull back $D^\chi u$ along $\widetilde{\chi}^\ast$. However, this operation is
{\em not always} well-defined. If we assume the condition
\begin{flalign}\label{prelim:WFpullbackcondition}
\pi\Big(s\mathrm{WF}(D^\chi u)\setminus \big( (\TT^\ast V\setminus {\bf 0})\times \{0\}\big)  \Big) \cap N_{\widetilde{\chi}}
=\emptyset\ ,
\end{flalign}
then $\chi_V^\ast u := \widetilde{\chi}^\ast D^\chi u \in \DD^\prime (U)\otimes \wedge^\bullet \bbR^{n}$
exists on account of the ordinary pullback theorem -- see Theorem 8.2.4 in Ref.~\onlinecite{Hormander}.
In fact, using Proposition \ref{propo:WFproperties}, the condition
\eqref{prelim:WFpullbackcondition} is equivalent to 
\begin{flalign}
\Big(\bigcup_{I\in\bbZ_2^{n}}\mathrm{WF}\big((D^\chi u)_I \big)\Big)\cap N_{\widetilde{\chi}}=\emptyset~.
\end{flalign}
By the ordinary pullback theorem this implies that all components $(D^\chi u)_I$ may be safely pulled back
along $\widetilde{\chi}^\ast$, and hence also $D^\chi u$.
Summing up, we have shown the following version of a pullback theorem for superdistributions.
\begin{theo}\label{theo:pullback}
Let $\chi : U^{m\vert n}\to V^{m^\prime\vert n^\prime}$ be a supermanifold morphism between two superdomains,
and consider the unique factorization $\chi_V^\ast = \widetilde{\chi}^\ast\circ D^\chi$ given in \eqref{eqn:pullbackcomponents}. 
Then the pullback map 
\begin{equation}
\chi_V^\ast : C^\infty(V)\otimes \wedge^\bullet\bbR^{n^\prime}\longrightarrow C^\infty(U)\otimes \wedge^\bullet \bbR^n
\end{equation}
has a unique continuous extension to those superdistributions $u\in \DD^\prime(V)\otimes \wedge^\bullet\bbR^{n^\prime}$
which satisfy the condition \eqref{prelim:WFpullbackcondition}.
\end{theo}

\begin{rem}
Another condition which would guarantee the existence of $\chi_V^\ast u$ is given by
\begin{flalign}\label{prelim:WFpullbackcondition2}
\pi\Big(s\mathrm{WF}(u)\setminus \big( (\TT^\ast V\setminus {\bf 0})\times \{0\}\big)  \Big) \cap N_{\widetilde{\chi}}
=\emptyset~.
\end{flalign}
In fact, using Proposition \ref{propo:WFproperties}, the condition
\eqref{prelim:WFpullbackcondition2} is equivalent to the strong condition 
$\mathrm{WF}(u_J) \cap N_{\widetilde{\chi}}=\emptyset$ for all
components $u_J$. 
Because differential operators preserve wavefront sets, it follows that 
\begin{multline}
\mathrm{WF}\big((D^\chi u)_I\big)\cap N_{\widetilde{\chi}} = \mathrm{WF}\Big(\sum_{J\in \bbZ^{n^\prime}_2} {(D^\chi)_I}^J u_J\Big)\cap N_{\widetilde{\chi}}\\
\subseteq \bigcup_{J\in \bbZ^{n^\prime}_2}\mathrm{WF}\big({(D^\chi)_I}^J u_J\big) \cap N_{\widetilde{\chi}}
\subseteq \bigcup_{J\in \bbZ^{n^\prime}_2}\mathrm{WF}(u_J) \cap N_{\widetilde{\chi}} =\emptyset
\end{multline}
for any $I$, which implies \eqref{prelim:WFpullbackcondition}.
Notice that the condition \eqref{prelim:WFpullbackcondition2} is 
much coarser than our condition \eqref{prelim:WFpullbackcondition}. Loosely speaking,
it does not take into account those components of $u$ which ``vanish algebraically under pullback'' 
due to the differential operator $D^\chi$. Let us illustrate this important point by an example:
Consider the  supermanifold morphism $\chi : \{\ast\} \to U^{m\vert n}$ which maps
a point into the superdomain $U^{m\vert n}$. Then 
\begin{flalign}
\chi_U^\ast :  C^\infty(U)\otimes \wedge^\bullet \bbR^n\longrightarrow \bbR~,~~f=\sum_{I\in \bbZ^{n}_2} f_I\,\theta^I\longmapsto f_{(0,\dots,0)}(\widetilde{\chi}(\ast))
\end{flalign}
is the mapping which ``forgets'' all higher components in the Grassmann algebra and evaluates the lowest component
at the point $\widetilde{\chi}(\ast)\in U$. 
We can clearly extend $\chi_U^\ast$ to {\em all} superdistributions $\DD^\prime(U)\otimes \wedge^\bullet \bbR^n$ 
with smooth lowest component $u_{(0,\dots,0)}\in C^\infty(U)$
by setting 
\begin{flalign}
u=\sum_{I\in\bbZ_2^n} u_I\,\theta^I\longmapsto u_{(0,\dots,0)}(\widetilde{\chi}(\ast))~.
\end{flalign}
Because $N_{\widetilde{\chi}} = \TT^\ast_{\widetilde{\chi}(\ast)} U$ is the cotangent space
at $\widetilde{\chi}(\ast)$, the condition \eqref{prelim:WFpullbackcondition2} is violated as soon as any
$u_I$ has a singularity at this point. In contrast, our condition \eqref{prelim:WFpullbackcondition} 
just involves the lowest component $u_{(0,\dots,0)}$ of the superdistribution, because
the matrix of differential operators reads as $D^\chi = \begin{pmatrix} 1 & 0 &\cdots & 0\end{pmatrix}$
and hence $D^\chi u = u_{(0,\dots,0)}$.
\end{rem}

In the remaining part of this section we specialize the result of Theorem \ref{theo:pullback}
to the important case where $\chi$ is the super diagonal mapping 
\begin{flalign}
\Delta : U^{m\vert n} \longrightarrow U^{m\vert n} \times U^{m\vert n}\simeq (U\times U)^{2m\vert 2n}~.
\end{flalign}
The underlying smooth map $\widetilde{\Delta} : U\to U\times U\,,~x\mapsto (x,x)$
is the diagonal map and $\Delta_{U\times U}^\ast : C^\infty(U\times U)\otimes 
\wedge^\bullet \bbR^{n}\otimes \wedge^\bullet\bbR^n 
\to C^\infty(U) \otimes \wedge^\bullet \bbR^n$ factorizes as
\begin{flalign}\label{eqn:factorizationproduct}
\Delta_{U\times U}^\ast = \widetilde{\Delta}^\ast \circ D^{\Delta}
=(\widetilde{\Delta}^\ast \otimes \id_{\wedge^\bullet\bbR^n})\circ (\id_{C^\infty(U\times U)} \otimes\mu)~,
\end{flalign}
where $\mu: \wedge^\bullet \bbR^{n} \otimes \wedge^\bullet \bbR^{n}\to \wedge^\bullet \bbR^{n}$ 
denotes the product in the Grassmann algebra $\wedge^\bullet \bbR^{n}$.
The normal set of $\widetilde{\Delta}$ can be characterized explicitly and it is given by
\begin{flalign}\label{eqn:normalproduct}
N_{\widetilde{\Delta}} = \Big\{\big((x,x),(k,-k)\big)\in \TT^\ast (U\times U)\,:\, (x,k)\in \TT^\ast U\Big\}~.
\end{flalign}
Given two superdistributions $u,v\in \DD^\prime(U)\otimes \wedge^\bullet \bbR^n$, their product (if it exists)
is given by $u\,v := \Delta_{U\times U}^\ast(u\otimes v)$. Expanding into components $u = \sum_{I\in\bbZ_2^n} u_I \, \theta^I$
and $v\in \sum_{J\in\bbZ_2^n} u_J \, \theta^J$, we obtain 
\begin{flalign}
u\otimes v = 
\sum_{I,J\in\bbZ_2^n} u_I\otimes v_J ~ (\theta^I\otimes \theta^J) \in 
\DD^\prime(U\times U)\otimes \wedge^\bullet \bbR^{n}\otimes \wedge^\bullet\bbR^n ~.
\end{flalign}
Due to the factorization \eqref{eqn:factorizationproduct}, the product  
of $u$ and $v$ (if it exists) is computed by first multiplying in the Grassmann algebra
\begin{flalign}
 D^{\Delta}\big(u\otimes v\big) = \sum_{I,J\in\bbZ_2^n} u_I\otimes v_J ~ (\theta^I\, \theta^J) 
\end{flalign}
and then pulling back the result component-wise via $\widetilde{\Delta}^\ast$, i.e.\
\begin{flalign}
u\,v:= \Delta_{U\times U}^\ast\big(u\otimes v\big) = \sum_{I,J\in\bbZ_2^n} \widetilde{\Delta}^\ast(u_I\otimes v_J) 
~ (\theta^I\, \theta^J) ~.
\end{flalign}
As a consequence of Theorem \ref{theo:pullback}, we have
\begin{cor}
 The product $u\,v\in \DD^\prime(U)\otimes \wedge^\bullet \bbR^n$ exists whenever 
\begin{flalign}\label{prelim:WFpullbackconditionMULT}
\pi\Big(s\mathrm{WF}\big(D^{\Delta}(u\otimes v) \big)\setminus \big( (\TT^\ast (U\times U)\setminus {\bf 0})\times \{0\}\big)  \Big)\cap N_{\widetilde{\Delta}}=\emptyset~,
\end{flalign}
or equivalently, whenever all components
$u_I,v_J\in\DD^\prime(U)$, for which $\theta^I\,\theta^J\neq 0$, can be multiplied in the sense
of ordinary distributions, cf.\ Theorem 8.2.4 in Ref.~\onlinecite{Hormander}.
\end{cor}
\begin{rem}
It is important to stress that the condition \eqref{prelim:WFpullbackconditionMULT}
in the corollary above {\em does not} impose conditions on the components $u_I$ and $v_J$
which multiply trivially on account of the Grassmann algebra structure, i.e.\ for which $\theta^I\,\theta^J=0$.
This is a clear advantage compared to the alternative (and much coarser) condition \eqref{prelim:WFpullbackcondition2}.
\end{rem}


\section{\label{sec:singularities}Singularities in supergeometric field theory}
In this section we apply the techniques developed in this paper
to analyze the singularities of the supergeometric
field theory introduced in Example \ref{ex:hyperbolic}. For simplifying 
our explicit computations, we consider only the case where $M=\bbR^3$ is the Minkowski spacetime,
i.e.\ we take the flat Lorentzian metric $g = \mathrm{diag}(1,-1,-1)$ on $M$.
In this case the equation of motion operator \eqref{eqn:EOM3D}
has constant coefficients and reads as
\begin{flalign}\label{eqn:EOM3DMink}
P = \begin{pmatrix}
m & 0 & -1\\
0 & \ii\gamma^\mu\partial_\mu +m & 0\\
g^{\mu\nu}\partial_\mu\partial_\nu & 0 & m
\end{pmatrix}~.
\end{flalign}
Let $u\in \DD^\prime(\bbR^3)\otimes \wedge^\bullet\bbR^2$
be any superdistribution satisfying $P u =0$. By Proposition \ref{propo:sWFpseudoaction},
the super wavefront set $s\mathrm{WF}(u)\subseteq \widehat{\mathcal{P}}^\ast \bbR^{3\vert 2}$   
of $u$ satisfies the equality
\begin{flalign}\label{tmp:wfcond}
\sigma_1(P)\big(s\mathrm{WF}(u)\big) = (\TT^\ast\bbR^3 \setminus{\bf 0}) \times\{0\}~,
\end{flalign}
where we also have used that $(\TT^\ast\bbR^3 \setminus{\bf 0}) \times\{0\}$ is the smallest possible
super wavefront set, cf.\ Corollary \ref{cor:minsWF}.
The equality \eqref{tmp:wfcond} is equivalent to the inclusion
\begin{flalign}\label{eqn:NcalP}
s\mathrm{WF}(u) \subseteq \mathcal{N}_P := 
\Big\{(x,k,\lambda) \in \widehat{\mathcal{P}}^\ast \bbR^{3\vert 2} \, :\, \sigma_1(P)(x,k)\big(\lambda\big)=0\Big\}~,
\end{flalign}
which follows by direct inspection of the left-hand-side of \eqref{eqn:NcalP} and using
\eqref{eqn:sWFpseudoaction}. Using the explicit form of the super principal symbol of \eqref{eqn:EOM3DMink},
we find the inclusion
\begin{flalign}\label{eqn:NcalPmore}
s\mathrm{WF}(u) \subseteq \Big((\TT^\ast\bbR^3 \setminus{\bf 0}) \times\{0\}\Big) 
\cup  \Big\{ \big(x,k,\phi + \psi\,\theta\big)\,:\, g^{\mu\nu}k_{\mu}k_{\nu}=0\, ,~\gamma^\mu k_\mu \psi =0 \Big\} ~,
\end{flalign}
where we have used the compact notation $\psi \, \theta := \psi_a\,\theta^a := \psi_1 \,\theta^1 + \psi_2\,\theta^2 $.
In words, \eqref{eqn:NcalPmore} tells us that all elements $(x,k,\lambda)\in s\mathrm{WF}(u)$
with nontrivial $\lambda\neq 0$ are such that $k$ is light-like.
Moreover, $\lambda = \phi + \psi\,\theta$ does not contain a quadratic $\theta$-term and the Fermionic
polarizations $\psi$ have to satisfy the Dirac-polarization constraint $\gamma^\mu k_\mu \psi =0$.
\sk

We next observe that the composition $\widetilde{P}\circ P$ of \eqref{eqn:EOM3DMink}
with the super (pseudo-)differential operator (of order $1$)
\begin{flalign}\label{eqn:EOM3Dtilde}
\widetilde{P}= \begin{pmatrix}
m & 0 & 1\\
0 & -\ii \gamma^\mu\partial_\mu +m & 0\\
-g^{\mu\nu}\partial_\mu\partial_\nu & 0 & m
\end{pmatrix}~
\end{flalign}
gives the component-wise Klein-Gordon equation
\begin{flalign}
\widetilde{P}\circ P = (g^{\mu\nu}\partial_\mu\partial_\nu + m^2)\,\id=: Q\,\id~.
\end{flalign}
In particular, each component $u_I$ of any $u$ satisfying $Pu=0$ satisfies the Klein-Gordon equation
$Q u_I =0 $, which entails the following inclusion
\begin{flalign}
\mathrm{WF}(u_I) \subseteq \Omega_Q := \Big\{(x,k)\in \TT^\ast \bbR^3\setminus {\bf 0} \, :\, g^{\mu\nu}k_\mu\,k_\nu =0\Big\}~,
\end{flalign}
for all component wavefront sets. By the standard propagation of singularities theorem (see Chapter 
26 in Ref.~\onlinecite{HormanderIV}), 
this implies that all $\mathrm{WF}(u_I)$ are invariant under the flow of the Hamiltonian vector field
\begin{flalign}
H_Q := \big\{\sigma_{2}(Q),\,\cdot\, \big\} =  2 g^{\mu\nu}k_\mu\,\partial_\nu : C^\infty(\Omega_Q)\longrightarrow C^\infty(\Omega_Q) ~,
\end{flalign}
i.e.\ any integral curve $c : \bbR \to \Omega_{Q}$ of $H_Q$ which satisfies $c(0)\in\mathrm{WF}(u_I) $ remains in
$\mathrm{WF}(u_I)$. In our example, any integral curve of $H_Q$ is of the form
\begin{flalign}\label{eqn:curveexplicit}
c : \bbR\longrightarrow \Omega_Q ~,~~s\longmapsto  \big(x^\mu + s\, 2 g^{\mu\nu}k_{\nu} , k_{\nu}\big)~,
\end{flalign}
for some $(x^\mu,k_\nu)\in \Omega_Q$. 
\sk

Following the ideas of Dencker \cite{Dencker},
we now shall study the propagation of polarizations in our example.
Given any integral curve $c : \bbR\to\Omega_Q$ of $H_Q$ as in \eqref{eqn:curveexplicit},
we consider the restriction of $\mathcal{N}_P$ given in \eqref{eqn:NcalP} to $c$, which gives
rise to a vector bundle
\begin{flalign}
\mathcal{N}_P\vert_{c }\longrightarrow \bbR~.
\end{flalign}
Using \eqref{eqn:NcalPmore}, we can compute its total space
\begin{flalign}\label{eqn:Nvertgamma}
\mathcal{N}_P\vert_{c} = \big\{\big(s ,\phi + \psi\,\theta\big)\,:\, \gamma^\mu k_\mu \psi =0\big\}~.
\end{flalign}
As the solution space of the Dirac-constraint $\gamma^\mu k_\mu \psi$ is one-dimensional (in $3$ dimensions),
the vector bundle $\mathcal{N}_P\vert_{c }\to \bbR$ is of rank two.
Following Definition 4.1 in Ref.~\onlinecite{Dencker}, a {\em Hamiltonian orbit} 
for our operator $P$ is a sub-line bundle $L\subseteq \mathcal{N}_P\vert_{c}$, where
$c$ is an integral curve as above and $L$ is spanned by a section $w \in \Gamma^\infty (\mathcal{N}_P\vert_{c})$
that satisfies $D_P^{} w =0$. Here $D_P^{}:= H_Q + \frac{1}{2}\{\sigma_{1}(\widetilde{P}),\sigma_1(P)\} 
+ \ii\,\sigma_{1}(\widetilde{P})\,\sigma_{0}^s(P)$
is a partial connection (cf.\ Equation (4.6) in Ref.~\onlinecite{Dencker}), where $\sigma_{0}^s(P)$ denotes the subprincipal
symbol of $P$. Clearly, the vector bundle $\mathcal{N}_P\vert_{c }$ can be spanned by the sections
$w \in \Gamma^\infty (\mathcal{N}_P\vert_{c})$ satisfying $D_P^{} w=0$.
In our example, we find that
\begin{flalign}
D_P^{} = \frac{\partial }{\partial s} + \ii\,m\, \begin{pmatrix}
0 & 0 & 1\\
0 & \gamma^\mu\,k_\mu & 0\\
0 & 0 & 0 
\end{pmatrix} : \Gamma^\infty (\mathcal{N}_P\vert_{c})\longrightarrow \Gamma^\infty (\mathcal{N}_P\vert_{c})~.
\end{flalign}
Notice that the connection coefficients (i.e.\ the second term in the expression above)
act trivially on the fibers of $\mathcal{N}_P\vert_{c}$ (this follows from \eqref{eqn:Nvertgamma}),
hence the expression for $D_P^{}$ simplifies to
\begin{flalign}
D_P^{} =\frac{\partial }{\partial s} :   \Gamma^\infty (\mathcal{N}_P\vert_{c})\longrightarrow \Gamma^\infty (\mathcal{N}_P\vert_{c})~.
\end{flalign}
Any Hamiltonian orbit in our example is therefore of the form
\begin{flalign}
\bbR\times \mathrm{span}_{\bbC}\big( \phi + \psi\,\theta\big)\subseteq \mathcal{N}_P\vert_{c}~,
\end{flalign}
for some $0\neq \phi + \psi\,\theta\in\wedge^\bullet\bbR^2$ satisfying $\gamma^\mu k_\mu \psi =0$.
\sk

Finally, we notice that $s\mathrm{WF}(u)$, for any $u$ satisfying $Pu=0$, 
is the union of such Hamiltonian orbits, i.e.\ the propagation of polarization result in Theorem 4.2 in Ref.~\onlinecite{Dencker}
remains valid in our supergeometric example. This follows from the fact that $Pu=0$
is equivalent to the component equations, for $u = \phi + \psi\,\theta + 
F\,\theta^1\,\theta^2\in\DD^\prime(\bbR^3)\otimes \wedge^\bullet\bbR^2$,
\begin{flalign}
m\,\phi = F\quad ,\qquad \ii\gamma^\mu\partial_\mu \psi +m\,\psi =0\quad ,\qquad  
g^{\mu\nu}\partial_\mu\partial_\nu \phi + m \,F=0~,
\end{flalign}
which can be decoupled into the Dirac equation $\ii\gamma^\mu\partial_\mu \psi +m\,\psi =0$
and the massive Klein-Gordon equation $g^{\mu\nu}\partial_\mu\partial_\nu \phi + m^2 \,\phi=0$. 
The absence of $F$-polarizations in the super wavefront set $s\mathrm{WF}(u)$ for
$F$ satisfying $m\,\phi = F$ follows from the discussion in Example \ref{ex:sWF}.
\sk

Combining Theorem \ref{theo:pullback} with this knowledge about propagation of singularities, 
it follows that any distributional solution to our supersymmetric field equation can be restricted to a 
Cauchy surface. The initial conditions for a well-posed supergeometric Cauchy problem, 
however, need to account for compatibility conditions between the Cauchy data in different degrees, 
see Refs.~\onlinecite{GKL,Grubb}. A detailed discussion will be given in ongoing work.


\section*{Acknowledgments}
We would like to thank Helmut Abels, Federico Bambozzi, Ulrich Bunke and Felix Finster for useful discussions.
C.D.\ is supported by the University of Pavia and would like to thank the Erwin Schr\"odinger Institute 
for the kind hospitality during the programme ``Modern theory of wave equations'', where part of this work has been done.
H.G.\ is supported by the ERC Advanced Grant HARG 268105.
S.M.\ is supported by the DFG Research Training Group GRK 1692 ``Curvature, Cycles, and Cohomology''
and by the COST Action MP1405 QSPACE through a short term scientific mission (STSM) to Edinburgh.
A.S.\ is supported by a Research Fellowship of the Deutsche Forschungsgemeinschaft (DFG, Germany)
and also would like to thank Felix Finster for the invitation to
Regensburg, where part of this research has been performed.


\end{document}